\documentclass[10pt]{article}

\usepackage{amsthm,amssymb,stmaryrd}
\usepackage{graphics}
\usepackage{amsfonts}
\usepackage{bbm}
\usepackage{tikz}
\usepackage[plainpages=false,pdfpagelabels=false,colorlinks=true,citecolor=blue,hypertexnames=false]{hyperref}
\usepackage{color}

\usepackage[switch]{lineno}

\usepackage{dsfont}
\usepackage{epsfig}
\usepackage{mathrsfs}
\usepackage{amsmath}
\usepackage{cases}

\newcommand{\bN} { {\mathbb{N}}}

\newcommand{\bQ} { {\mathbb{Q}}}
\newcommand{\bZ} { {\mathbb{Z}}}

\def\res{\operatorname{res}}

\makeatletter
\newcommand{\myitem}[1]{%
	\item[(#1)]\protected@edef\@currentlabel{#1}%
}
\makeatother

\def\eatspace#1{#1}
\def\step#1#2{\par\kern1pt\hangindent#2em\hangafter=1\noindent\rlap{\small#1}\kern#2em\relax\eatspace}

\let\set\mathbb
\def\<#1>{\langle#1\rangle}

\def\val{\operatorname{val}}

\def\disc{\operatorname{Disc}}
\def\diag{\operatorname{diag}}
\def\im{\operatorname{im}}

\def\lt{\operatorname{lt}}
\def\codim{\operatorname{codim}}
\def\Tr{\operatorname{Tr}}

\def\S{\mathcal{S}}

\newtheorem{theorem}{Theorem}
\newtheorem{prop}[theorem]{Proposition}

\newtheorem{lemma}[theorem]{Lemma}

\newtheorem{algorithm}[theorem]{Algorithm}

\newtheorem{defi}[theorem]{Definition}
\newtheorem{example}[theorem]{Example}

\makeatletter \@addtoreset{equation}{section}
\begin{document}
\title{Lazy Hermite Reduction and Creative Telescoping for Algebraic~Functions
\thanks{S.\ Chen was partially supported by the NSFC
	grants 11871067, 11688101, the Fund of the Youth Innovation Promotion Association, CAS,
	and the National Key Research and Development Project 2020YFA0712300.  L.\ Du
	was supported by the NSFC grant 11871067 and the Austrian FWF grant P31571-N32. M. \ Kauers was supported by the
	Austrian FWF grant P31571-N32.}}

\author{\bigskip Shaoshi Chen$^{a, b}$, Lixin Du$^{a, b, c}$, Manuel Kauers$^c$\\
$^a$KLMM, Academy of Mathematics and Systems Science,\\ Chinese Academy of Sciences,\\ Beijing 100190, China\\
$^b$School of Mathematical Sciences,\\ University of Chinese Academy of Sciences,\\ Beijing 100049, China\\
$^c$Institute for Algebra, Johannes Kepler University,\\ Linz, A4040,  Austria\\
{\sf schen@amss.ac.cn,  dulixin17@mails.ucas.ac.cn}\\
{\sf manuel.kauers@jku.at}
}

\maketitle

\begin{abstract}
Bronstein's lazy Hermite reduction is a symbolic integration technique that
reduces algebraic functions to integrands with only simple poles without the prior computation of an integral basis.
We sharpen the lazy Hermite reduction by combining it with the polynomial reduction to solve
the decomposition problem of algebraic functions. The sharpened reduction is then used
to design a reduction-based telescoping algorithm for algebraic functions in two variables.
\end{abstract}

\section{Introduction}

The integration problem for algebraic functions has a long history
that can be traced back at least to the work of Euler and others on elliptic integrals~\cite{Abhyankar1976}.
In 1826, Abel initiated the study of general integrals of algebraic functions, which are now called abelian integrals~\cite{Abel1826}.
Liouville in 1833 proved that if the integral of an algebraic function is an elementary function, then it must be the sum of an algebraic function and
a linear combination of logarithms of algebraic functions (see~\cite[Chapter IX]{Lutzen1990} for a detailed historical overview). In 1948, Ritt presented an algebraic approach
to the problem of integration in finite terms in his book~\cite{Ritt1948}.
Based on Liouville's theorem and
some developments in differential algebra~\cite{Ritt1950}, Risch in 1970 finally solved this
classical problem by giving a complete algorithm~\cite{Risch1970}. After Risch's work, more efficient algorithms
have been given due to the emerging developments in symbolic computation~\cite{Davenport1981, Trager84, Bronstein1990, BronsteinBook}.
Passing from indefinite integration to definite integration with parameters, the central problem shifts to finding
linear differential equations satisfied by the integrals of algebraic functions with parameters. In this direction, the first work was started by Picard in~\cite{picard33} and
later a systematical method was developed by Manin in his work on proving the function-field analogue of Mordell's conjecture~\cite{Manin1958}.
In the 1990s, another powerful method was developed by Almkvist and Zeilberger~\cite{Almkvist1990} by including
the trick of differentiating under the integral sign in the framework of creative telescoping~\cite{Zeilberger1991}.

In a given differential ring $R$ with the derivation $D$, one can ask two fundamental problems: one is the \emph{integrability problem}, i.e., deciding
whether a given element $f\in R$ is of the form $D(g)$ for some $g\in R$, if such a $g$ exists, we say that $f$ is \emph{integrable} in $R$;
another is the \emph{decomposition problem}, i.e., decomposing a given
element $f\in R$ into the form $D(g) + r$ with $g, r\in R$ and $r$ is minimal in certain sense and $r=0$ if and only if $f$ is integrable in $R$. For algebraic functions, the integrability problem was studied
by Liouville in his two memoirs~\cite{Liouville1833a, Liouville1833b}. Hermite reduction~\cite{Hermite1872} solves the decomposition problem
for rational functions. Trager in his thesis~\cite{Trager84} extended Hermite reduction to the algebraic case.
His algorithm requires the computation of an integral bases in the beginning. In order to avoid this expensive step, Bronstein~\cite{bronstein98a} introduced the lazy
Hermite reduction that partially solves the decomposition problem.  The first contribution in this paper is a
sharpened version of the lazy Hermite reduction. We combine the lazy Hermite reduction with a further reduction, namely the polynomial reduction, in order to
solve the decomposition problem completely.

When the differential ring $R$ is equipped with two derivations $D_1, D_2$, one can also
consider the \emph{creative telescoping} problem: for a given element $f\in R$, decide whether there exist $c_0, \ldots, c_r \in R$, not all zero,  such that $D_2(c_i) =0$ for all $i \in \{0, \ldots, r\}$ and
\[c_r D_1^r(f) + \cdots + c_0 f = D_2(g)\quad \text{for some $g\in R$}.\]
The operator $L = c_r D_1^r + \cdots + c_0$ if exists is called a \emph{telescoper} for~$f$.
For every algebraic function there exists such a telescoper, and many construction algorithms have been
developed in~\cite{Manin1958, Zeilberger1990, chen12d, bostan13b, chen16a}. Our second contribution is
an adaption of the reduction-based approach from~\cite{chen16a} {using the sharpened} lazy Hermite reduction.

The remainder of this paper is organized as follows.
First we recall Bronstein's idea of lazy Hermite reduction {in Section~\ref{SECT:lazy}}.
Instead of an integral basis, it uses a so-called ``suitable basis'' of the function field.
In Section~\ref{SECT:suitable} we have a closer look at these bases.
After developing the polynomial reduction in Section~\ref{SECT:polyred}, we will present the telescoping
algorithm for algebraic functions based on the sharpened lazy Hermite reduction in Section~\ref{SECT:tele}.
We conclude our paper by some experimental comparisons among several telescoping algorithms in Section~\ref{SECT:experiments}.


\section{Lazy Hermite Reduction} \label{SECT:lazy}

Trager's generalization of Hermite reduction to algebraic functions works as follows~\cite{Trager84,geddes92,bronstein98,chen16a}. Let $K$ be a field of characteristic zero and {$m\in K(x)[y]$ be an irreducible polynomial
	over $K(x)$. Then $A=K(x)[y]/\<m>$ is an algebraic extension of $K(x)$}. When there is no ambiguity, we also use $y$ to represent the element $y+\<m>$ in $A$, which can be viewed as a root of $m$. Let $W=(\omega_1,\dots,\omega_n)$
be an integral basis of~$A$.  Let $f=\frac1{uv^d}\sum_{k=1}^n a_k\omega_k\in A$ with $d>1$ and
$u,v,a_1,\dots,a_n\in K[x]$ such that $\gcd(u,v)=\gcd(v,v')=\gcd(v,a_1,\dots,a_n)=1$. We seek
$b_1,\dots,b_n,c_1,\dots,c_n\in K[x]$ such that for $g=\frac1{v^{d-1}}\sum_{k=1}^n b_k\omega_k$
and $h=\frac1{uv^{d-1}}\sum_{k=1}^n c_k\omega_k$ such that
\[
f = g' + h.
\]
The $g$ in this equation can be found by solving a certain linear system over $K[x]/\<v>$,
and once $g$ is known, $h$~can be computed as $h=f-g'$.
Let $e\in K[x]$ and $M\in K[x]^{n\times n}$ be such that $eW'=MW$, and assume (without loss
of generality) that $e\mid uv$. Then the coefficient vector $b=(b_1,\dots,b_n)$ of $v^{d-1}g$
satisfies
\begin{equation}\label{eq:sys}
b \bigl(uve^{-1}M - (k-1) u v' I_n\bigr) \equiv (a_1,\dots,a_n) \bmod v
\end{equation}
and using that $W$ is an integral basis, it can be shown that this linear system has a
unique solution, see~\cite{Trager84,chen16a} for further details. Applying the process repeatedly,
we can eliminate all multiple poles from the integrand, i.e., we can find $g$ and $h$ such
that $f=g'+h$ and $h=q^{-1}\sum_{k=1}^n p_k\omega_k$ for some polynomials $p_1,\dots,p_n,q$
with $q$ squarefree.

If $W$ is not an integral basis, the linear system~\eqref{eq:sys} may or may not have a unique
solution.

\begin{example}
	Let $m=y^2-x$ and $f=\frac{y}{(x+1)x^2}$.
	\begin{enumerate}
		\item For $W=(x,xy)$ we have $e=2x$ and $M=\begin{pmatrix}2&0\\0&3\end{pmatrix}$. Applying the reduction to $f=\frac1{(x+1)x^3} xy$ leads
		to the linear system
		\[
		b\begin{pmatrix}
		-2(x+1) & 0 \\ 0 & -(x+1)
		\end{pmatrix} \equiv (0,2)\bmod x
		\]
		which has a unique solution.
		\item For $W=(x,y)$ we have $e=2x$ and $M=\begin{pmatrix}1&0\\0&1\end{pmatrix}$.
		Applying the reduction to $f=\frac1{x^2(x+1)}y$ leads to the linear system
		\[
		b \begin{pmatrix}
		0 & 0 \\
		0 & -(x+1)
		\end{pmatrix} \equiv (0, 2)\bmod x,
		\]
		which is solvable, but not uniquely.
		\item
		For $W=(x,(x+1)y)$ we have $e=2x(x+1)$ and $M=\begin{pmatrix}2(x+1)&0\\0&3x+1\end{pmatrix}$.
		Applying the reduction to
		\[
		f=\frac1{x^2(x+1)^2}(x+1)y
		\]
		leads to the linear system
		\[
		b \begin{pmatrix}
		-2x & 0 \\ 0 & -(x+1)
		\end{pmatrix} \equiv (0,2)\bmod x(x+1)
		\]
		which has no solutions.
	\end{enumerate}
\end{example}

Note that none of the bases above in the example above is an integral basis. However, all the bases
consist of integral elements and have the property that $e$ is squarefree.
Bronstein~\cite{bronstein98a} calls such a basis a \emph{suitable basis} and observes that whenever
we apply Hermite reduction to a suitable basis and find that the linear system~\eqref{eq:sys}
has no solution, then we can construct from any unsolvable system an integral element of $A$
that does not belong to the $K[x]$-module generated by the elements of~$W$. We can then
replace $W$ by a suitable basis of an enlarged $K[x]$-module which also includes $A$ and
proceed with the reduction.

\begin{example}\label{ex:lhr}
	We continue the previous example.
	\begin{enumerate}
		\item For $W=(x,xy)$, no basis update is needed because the linear system has a unique solution.
		\item For $W=(x,y)$, the right kernel element $(1,0)$ of the matrix in the linear system
		translates into the new integral element $x+1$, which does not belong to the $K[x]$ module
		generated by $x$ and $y$ in~$A$. A basis of the module generated by $x$, $y$, and $x+1$
		is $(1,y)$.
		\item For $W=(x,(x+1)y)$, from the lack of solutions of the linear system it can be deduced
		that $xy$ is an integral element not belonging to the module generated by $x$ and $y$ in~$A$.
		A basis of the module generated by $x$, $y$, and $xy$ is $(x,y)$.
	\end{enumerate}
\end{example}

Starting from a basis $W$ consisting of integral elements, there can be at most finitely many
basis updates before we reach an integral basis. Therefore, it takes at most finitely many
basis updates (and possibly fewer than needed for reaching an integral basis) to complete the
reduction process. This variant of Hermite reduction, which avoids the potentially
expensive computation of an integral basis at the beginning, is called lazy Hermite reduction.
Its final result is a suitable basis $\tilde W$ of $A$ and $g,h\in A$
such that $f=g'+h$ and the coefficients of $h$ with respect to $\tilde W$ are rational functions
with a squarefree common denominator.
In the examples above, we may find $\tilde W=(1,y)$, $g=-\frac{2y}{x}$, and $h=-\frac{y}{x(x+1)}$.

One of the key features of Hermite reduction is that we can decide the integrability problem.
For example, if we write a rational function $f\in K(x)$ in the form $f=g'+h$ for some $g,h\in K(x)$
where $h$ has a squarefree denominator and numerator degree less than denominator degree, then $f$
admits an integral in $K(x)$ if \emph{and only if} $h=0$.
Trager generalizes this criterion to algebraic functions as follows. By a change of variables, he
first ensures that the integrand $f$ has a double root at infinity. Then he performs Hermite reduction
with respect to an integral basis that is normal at infinity. If this gives $g,h$ such that $f=g'+h$,
then $f$ is integrable in $A$ if \emph{and only if} $h=0$~\cite{Trager84,chen16a}.

Unfortunately, this criterion does not extend to the lazy version of Hermite reduction. Even if we
produce a double root of the integrand at infinity and make the suitable basis normal at infinity
(which amounts to a local integral basis computation that we actually would prefer to avoid altogether),
a nonzero remainder $h$ does not imply that $f$ is not integrable.

\begin{example}\label{ex:doulbe root at infy}
	Let $m=y^2-x$ and $f=\frac{y}{x^3}$. For $W=(1,(x^2+1)y)$ we have $e=2x(x^2+1)$ and $M=\begin{pmatrix}0&0\\0&5x^2+1\end{pmatrix}$. The lazy Hermite reduction finds $g=-\frac{2 (1 + x^2) y}{3 x^2}$ and $h=\frac{y}{3x}$.
	{Here} $f=\left(-\frac{2y}{3x^2}\right)^\prime$ is integrable and has a root of order $\geq2$ at infinity, but the remainder $h$ is nonzero.
\end{example}

Somewhat surprisingly, Bronstein does not address this issue in his report~\cite{bronstein98a}.
In the following sections, we develop a fix using the technique of polynomial reduction.

\section{Suitable bases}\label{SECT:suitable}

Let $A=K(x)[y]/\<m>$ with $m\in K[x,y]$ being an irreducible polynomial over $K(x)$. Let $\bar K$ be the algebraic closure of $K$. If $n=\deg_y(m)$, there are $n$ distinct solutions in the field
\[
\bar K\<\<x-a>>:=\bigcup_{r\in\bN\setminus\{0\}}\bar K((\,(x-a)^{1/r}))
\]
of formal Puiseux series around $a\in \bar K$. There are also $n$ distinct solutions in the field
\[
\bar K\<\<x^{-1}>>:=\bigcup_{r\in\bN\setminus\{0\}}\bar K((x^{-1/r}))
\]
of formal Puiseux series around~$\infty$. For a fixed $a\in\bar K\cup\{\infty\}$, let $y_1,\ldots,y_n$ be all $n$ roots of $m$ in $\bar K\<\<x-a>>$ (or $\bar K\<\<x^{-1}>>$ if $a=\infty$).
There are $n$ distinct $K(x)$-embeddings $\sigma_1,\ldots,\sigma_n$ of $A$ into $\bar K\<\<x-a>>$ (or $\bar K\<\<x^{-1}>>$ if $a=\infty$) such that $\sigma_i(f(y))=f(y_i)$ for any $f\in A$.
Then for each $a\in \bar K\cup\{\infty\}$, we can associate $f\in A$ with $n$ series $\sigma_i(f)$ for $i=1,\ldots,n$.
Moreover, if we equip the fields $A$, $\bar K \<\<x-a>>$ and $\bar K\<\<x^{-1}>>$ with natural differentiations with respect to~$x$, then each embedding $\sigma_i$ is a differential homomorphism, i.e., $\sigma_i(f^\prime)=\sigma_i(f)^\prime$ for any $f\in A$.

A nonzero Puiseux series around $a\in \bar K$ can be written in the form
\[P=\sum_{i\geq0}c_i(x-a)^{r_i},\]
where $c_i\in\bar K,c_0\neq0$ and $r_i\in\bQ$. The {\em valuation map} $\nu_a$ on $\bar K\<\<x-a>>$ is defined by $\nu_a(P)=r_0$ if $P$ is nonzero and $\nu_a(0)=\infty$. Replacing $x-a$ by $\frac{1}{x}$, we get the {\em valuation} map $\nu_{\infty}$ on $\bar K\<\<x^{-1}>>$. A series $P$ in $\bar K\<\<x-a>>$ or $\bar K\<\<x^{-1}>>$ is called {\em integral} if its valuation is nonnegative. The value function $\val_a: A\to \bQ\cup\{\infty\}$ is defined by
\[\val_a(f):=\min_{i=1}^n\nu_a(\sigma_i(f))\]
for any $f\in A$. An element $f\in A$ is called {\em (locally) integral} at $a\in\bar K\cup\{\infty\}$ if $\val_a(f)\geq0$, i.e., every series associated to $f$ (around $a$) is integral. The element $f\in A$ is called {\em (globally) integral} if $\val_a(f)\geq0$ at every $a\in \bar K$ (``at all finite places''), i.e., $f$ is locally integral at every $a\in\bar K$.
A basis of the $K[x]$-module of all integral elements of $A$ is called an integral basis of~$A$.
The elements of $A$ that are locally integral at some point $a\in K$ form a $K(x)_a$ module, where $K(x)_a$ is the subring of $K(x)$ consisting of all rational functions which do not have a pole at~$a$.
In the case $a=\infty$, $K(x)_a$ is the ring of all rational functions $p/q$ with $\deg_x(p)\leq\deg_x(q)$.
A basis of the $K(x)_a$-module of locally integral elements of $A$ is called a local integral basis at~$a$.

For a series $P\in\bar K\<\<x-a>>$, the smallest positive integer $r$ such that $P\in \bar K((\,(x-a)^{1/r}))$ is called the {\em ramification index} of $P$. If $r>1$, the series $P$ is said to be {\em ramified}. For an element $f\in A$, a point $a\in\bar K$ is called a {\em branch point} of $f$ if one of the series associated  to $f$ around $a$ is ramified.

Let $W=(\omega_1,\ldots,\omega_n)$ be a $K(x)$-vector space basis of~$A$. Throughout this section, let $e\in K[x]$ and $M=((m_{i,j}))_{i,j=1}^n\in K[x]^{n\times n}$ be such that
\begin{equation*}
e W^\prime=MW
\end{equation*}
and $\gcd(e,m_{1,1},m_{1,2},\ldots,m_{n,n})=1$.
As already mentioned in the previous section, $W$ is a called a \emph{suitable basis} if $e$ is squarefree and {$\omega_i$'s are integral for each $i$}. Every integral basis is a suitable basis, see~\cite[Lemma 3]{chen16a}.
Now we explore some further properties of such bases.

\begin{lemma}\label{lem:e branch}
	Let $W$ be an integral basis of~$A$. Let $e\in K[x]$ and $M \in K[x]^{n\times n}$ be such that $e W^\prime=M W$. If $a\in\bar K$ is a root of $e$, there exists $\omega\in W$ such that $a$ is a branch point of $\omega$.
\end{lemma}
\begin{proof}
	Let $W=(\omega_1,\ldots,\omega_n)$ and $M=((m_{i,j}))_{i,j=1}^n\in K[x]^{n \times n}$, and let $a$ be a root of~$e$. By $\gcd(e,m_{1,1},\ldots,m_{n,n})=1$, there is $i\in\{1,\ldots,n\}$ such that
	\begin{equation*}
	\omega_i^\prime=\frac{1}{e}\sum_{j=1}^nm_{i,j}\omega_j,
	\end{equation*}
	where $a$ is not a common root of $m_{i,1},\ldots,m_{i,n}$.
	
	From the above expression of $\omega_i^\prime$, we get $\omega_i^\prime$ does not belong to the module generated by $W$ over $K[x]$. Since $W$ is a local integral basis at $a$, it follows that $\omega_i$ is not locally integral at $a$. Thus  $\val_a(\omega_i^\prime)<0$.
	
	If $a$ were not a branch point of $\omega_i$, then all Puiseux series at~$a$ associated to $\omega_i$ were power series, and then all Puiseux series at $a$ associated to $\omega_i'$ were power series,
	implying $\val_a(\omega_i')\geq0$. As we have seen before that $\val_a(\omega_i')<0$, it follows that $a$ must be a branch point.
\end{proof}

In order to give a converse of Lemma~\ref{lem:e branch}, we consider the series associated to an algebraic function. For a ramified Puiseux series $P=\sum_{i\geq0}c_i(x-a)^{r_i}\in\bar K\<\<x-a>>$, let
\[\delta(P):=\min\{r_i\mid\text{$r_i\in\bQ\setminus\set Z$, $i\geq0$}\}\]
be the minimal fractional exponent of $P$. Define $\delta(P)=\infty$ if the series $P$ is not ramified. Then $\delta(P^\prime)=\delta(P)-1$. Similar as the valuation of a series, the function $\delta$ satisfies \[\delta(P+Q)\geq\min\{\delta(P),\delta(Q)\}\]
for any $P,Q\in \bar K\<\<x-a>>$.
\begin{lemma}\label{lem:e2 branch}
	Let $W$ be a $K(x)$-vector space basis of~$A$. Let $e\in K[x]$ and $M\in K[x]^{n\times n}$ be such that $e W^\prime=M W$. Let $a\in \bar K$. If there exists some $\omega\in W$ such that $a$ is a branch point of $\omega$, then $a$ is a root of~$e$.
\end{lemma}
\begin{proof}
	Suppose that $a$ is a branch point of some $\omega\in W$. It implies that for such an element $\omega$, there is  a $K(x)$-embedding $\sigma$ of $A$ into the field of Puiseux series (around $a$) such that the series $\sigma(\omega)$ is ramified. Let $r=\min\{ \delta(\sigma(\omega))| \omega\in W\}$. Then $r\in \bQ\setminus\bZ$. Choose an element $\omega_i\in W$ such that $\delta(\sigma(\omega_i))=r$. Then
	$\sigma(\omega_i)$ must be ramified. After differentiating the series $\sigma(\omega_i)$, its minimal fractional exponent decreases strictly by~$1$. This means
	\[\delta(\sigma(\omega_i)^\prime)=\delta(\sigma(\omega_i))-1=r-1.\]
	
	Let $M=((m_{i,j}))_{i,j=1}^n\in K[x]^{n\times n}$. Since $\sigma$ is a differential homomorphism and a $K(x)$-embedding, we have \[\sigma(\omega_i)^\prime=\sigma(\omega_i^\prime)=
	\frac{1}{e}\sum_{j=1}^nm_{i,j}\sigma(\omega_j).\]  After multiplying by a polynomial, the minimal fractional exponent of a series will not decrease. So if $a$ is not a root of $e$, then \[\delta(\sigma(\omega_i)^\prime)\geq\min_{j=1}^n\delta(m_{i,j}\sigma(\omega_j))\geq\min_{j=1}^n\delta(\sigma(\omega_j))=r.\]
	This leads to a contradiction. Thus $a$ must be a root of $e$.
\end{proof}

{We now show that} the polynomial $e$ does not depend on the choice of the basis of $A$ but only on the $K[x]$-submodule it generates. Let $U$ and $V$ be two $K(x)$-vector space bases of $A$. Let $e_u, e_\nu\in K[x]$ and $M_u, M_\nu\in K[x]^{n\times n}$ be such that $e_uU^\prime=M_uU$ and $e_\nu V^\prime=M_\nu V$.
Suppose that $U$ and $V$ generate the same submodule of $A$ over $K[x]$. Then there exists a matrix $T\in K[x]$ such that $U=TV$ and $T$ is an invertible matrix over $ K[x]$. Taking derivatives, we get
\[U^\prime=\left(T^\prime T^{-1}+T\frac{1}{e_\nu}M_\nu T^{-1}\right)U=\frac{1}{e_u}M_uU.\]
Since $T,T^{-1}\in  K[x]^{n\times n }$,  we have $e_u$ divides $e_\nu$. Similarly the fact that $V=RU$ with $R=T^{-1}\in K[x]^{n\times n}$ implies that $e_\nu$ divides $e_u$. Thus $e_u=e_\nu$ when $e_u,e_\nu$ are monic.

\begin{lemma}\label{lem:e det(T)}
	Let $W$ be a suitable basis and $U$ be an integral basis of~$A$. Let $W=TU$ with $T\in K[x]^{n\times n}$. Let $e\in K[x]$ and $M\in K[x]^{n\times n}$ be such that $e W^\prime=M W$. If $a\in\bar K$ is a root of $\det(T)$, then $a$ is a root of~$e$. That means if $W$ is not a local integral basis at $a\in \bar K$, then $a$ is a root of~$e$.
\end{lemma}
\begin{proof}
	First we shall make a change of bases. Consider the Smith normal form of the matrix $T$. This means there are two matrices $P,Q\in  K[x]^{n \times n}$ such that both $P$ and $Q$ are invertible over $K[x]$ and $PTQ=\Lambda$ for some diagonal matrix $\Lambda\in K[x]^{n\times n}$. Then $W=P^{-1}\Lambda Q^{-1}U$ and hence $PW=\Lambda (Q^{-1}U)$. Replacing $PW$ and $Q^{-1}U$ by $W$ and $U$, respectively, we may assume
	\[W=TU\]
	with $T=\diag(r_1,\ldots,r_n)\in K[x]^{n\times n}$. This operation does not change the module generated by $W$ or $U$, respectively. Note that $\det(PTQ)$ and $\det(T)$ are equal up to a unit in $K$. It suffices to prove the result for such special bases $W$ and $U$.
	
	Let $e_u\in K[x]$ and $M_u=((a_{i,j}))_{i,j=1}^n\in K[x]^{n\times n}$ be such that $e_u  U^\prime=M_u U$. Differentiating both sides of  $W=TU$ yields that
	\begin{equation}\label{eq:d(W=TU)}
	W^\prime=\left(T^\prime T^{-1}+T\frac{1}{e_u}M_uT^{-1}\right)W=\frac{1}{e}MW.
	\end{equation}
	Substituting $T=\diag(r_1,\ldots,r_n)$ and $M_u=((a_{i,j}))_{i,j=1}^n$, we get the the $i$-th
	diagonal entry of $\frac{1}{e}M$ is $\frac{r_i^\prime}{r_i} + \frac{a_{i,i}}{e_u}$.
	
	Let $a\in \bar K$ be a root of $\det(T)$. Since $\det(T)=r_1r_2\cdots r_n$, there is $i\in\{1,2,\ldots,n\}$ such that $a$ is a root of $r_i$. If $a$ is not a root of $e_u$, then $a$ must be a pole of the entry $\frac{r_i^\prime}{r_i}+\frac{a_{i,i}}{e_u}$. Since $e$ is a common multiple of the denominator of all the entries of $\frac{1}{e}M$, we have $a$ is a root of $e$. Now suppose that $a$ is a root of $e_u$.  Since $U$ is an integral basis, Lemma~\ref{lem:e branch} implies that there is $u\in U$ such that $a$ is a branch point of $u$. Then one of the series associated to $u$ is ramified. In other words, such a series has at least one fractional exponent. Since $W=TU$ for some invertible matrix $T$ in $K(x)^{ n\times n}$,  there is $\omega\in W$ such that one of the series associated to $\omega$ is also ramified. Therefore, $a$ is a branch point of $\omega$. By Lemma~\ref{lem:e2 branch}, $a$ is root of $e$.
\end{proof}

Since the inverse of the matrix $T$ is $T^{-1}=\frac{1}{\det(T)}T^*$, where
$T^*$ is the adjoint matrix of $T$, the least common multiple of the denominator
of the entries of $T^{-1}$ is bounded by $\det(T)$. By investigating
Equation~\eqref{eq:d(W=TU)}, we see a possible root of $e$ must come from
$\det(T)$ and $e_u$. Combining the last paragraph in the argument of Lemma~\ref{lem:e det(T)}, we get those roots of $\det(T)$ and $e_u$ are exactly roots of $e$. In
particular, when $W$ is a suitable basis, the corresponding polynomial $e$ is
the squarefree part of the product $\det(T)e_u$. Therefore, if the submodule
generated by a suitable basis is larger, then $e$ is smaller.

\section{Polynomial Reduction} \label{SECT:polyred}

Polynomial reduction is a postprocessing step for Hermite reduction which was first introduced for
hyperexponential terms~\cite{bostan13a} and has later been formulated for algebraic and fuchsian D-finite
functions~\cite{chen16a,chen17a}.
For the latter cases, like for Trager's criterion, integral bases that are normal at infinity are
employed. Our goal in this section is to relax this requirement to suitable bases, so as to obtain
a version of polynomial reduction which can serve as a natural continuation of the lazy Hermite
reduction process and provides the feature that the final remainder is zero if and only if the
integrand is integrable.

Let $h\in A$ be the remainder of lazy Hermite reduction (see Section~\ref{SECT:lazy}) with respect to a suitable basis $W=(\omega_1,\dots,\omega_n)$. Then $h$ can be written in the form
\begin{equation}\label{eq:lazy remainder}
h=\sum_{i=1}^n\frac{h_i}{de}\omega_i
\end{equation}
with $h_i,d\in K[x]$ such that $\gcd(d,e)=\gcd(h_1,\ldots,h_n,d)=1$ and $d$ is squarefree. If $h$ is integrable in $A$, we shall prove that $d$ is a constant in~$K$.
When $W$ is an integral basis, this result was proved in~\cite[Lemma 9]{chen16a}. The following lemma is a local version.
\begin{lemma}\label{lem:9 local}
	Let $h\in A$ be in the form~\eqref{eq:lazy remainder}. If $h$ is integrable in $A$ and $W$ is a local integral basis at $a\in \bar K$, then $d$ has no root at $a$.
\end{lemma}
\begin{proof}
	Suppose $h$ is integrable in $A$. Then there exists $H=\sum_{i=1}^nb_i \omega_i\in A$ with $b_i\in K(x)$ such that $h=H^\prime$. It suffices to show that every $b_i$ has no pole at~$a$. Otherwise $H$ has a pole at~$a$, because $W$ is a local integral basis at~$a$. Then $h$ has at least a double pole at~$a$. This contradicts the fact that $d,e$ are squarefree. Thus $d$ has no root at~$a$.
\end{proof}

\begin{theorem}\label{thm:lemma 9}
	Let $h\in A$ be in the form~\eqref{eq:lazy remainder}. If $h$ is integrable in $A$, then $d$ is in~$K$.
\end{theorem}
\begin{proof}
	Suppose that $h$ is integrable in~$A$. In order to show $d$ is a constant, we show that for any $a\in \bar K$, $a$ is not a root of~$d$. If $W$ is a local integral basis at~$a$, then the conclusion follows from Lemma~\ref{lem:9 local}. Now we assume that $W$ is not a local integral basis at~$a$. By Lemma~\ref{lem:e det(T)}, we know that $a$ is a root of~$e$. Since $\gcd(d,e)=1$, it follows that $a$ is not a root of~$d$.
\end{proof}

{To further reduce the lazy Hermite remainder, we give an upper bound for the denominator of its integral if $h$ is integrable in $A$. This bound does not depend on the integrand $h$, but only depends on the discriminant of a suitable basis in its representation.}

Recall that the {\em discriminant} of a tuple $W=(\omega_1,\ldots,\omega_n)$ of elements of $A$ is defined by the determinant
\[\disc(W)=\det(\Tr(\omega_i\omega_j)),\]
where $\Tr$ is the trace map from $A$ to $K(x)$. If the $\omega_i$'s are integral functions, then their traces are polynomials, and thus {the discriminant} is a polynomial. If $W=TU$ where $T\in K[x]^{n\times n}$ is the change of basis matrix, then $\disc(W)=\disc(U)\det(T)^2$.
\begin{lemma}\label{lem:bound}
	Let $h\in A$ be in the form~\eqref{eq:lazy remainder}. Let $U$ be an integral basis of $A$ and let $T\in K[x]^{n\times n}$ be such that $W=TU$. If $h$ is integrable in $A$, i.e., there exists $u\in K[x]$ and $q=(q_1,\ldots,q_n)\in K[x]^n$ such that
	\[h=\left(\sum_{i=1}^n\frac{q_i}{u}\omega_i\right)^\prime\]
	and $\gcd(q_1,\ldots,q_n,u)=1$,
	then $u$ divides $\det(T)$. Hence $u^2$ divides $\disc(W)/\disc(U)$.
\end{lemma}
\begin{proof}
	{Since $h$~has a squarefree denominator with respect to~$W$ and $W=TU$ with $T\in K[x]^{n\times n}$, it follows that $h$~has a squarefree denominator with respect to $U$. This means $h$ is also a remainder with respect to the integral basis~$U$.} Assume that $h$ is integrable in~$A$.  Then Lemma~\ref{lem:9 local} implies that { $h=(\sum_{i=1}^n a_i u_i)^\prime$ with $a_i\in K[x]$.} Write $U=\frac{1}{r}RW$ with $r\in K[x]$, $R=((b_{i,j}))_{i,j=1}^n\in K[x]^{n\times n}$ and $\gcd(r,b_{1,1},b_{1,2},\ldots,b_{n,n})=1$. Then
	\[\sum_{i=1}^na_iu_i=\sum_{i=1}^na_i\frac{1}{r}\sum_{j=1}^nb_{i,j}\omega_j=\sum_{j=1}^n\frac{1}{r}\left(\sum_{i=1}^na_ib_{i,j}\right)\omega_j.\]
	Thus $u$ divides $r$ because $a_i,b_{i,j}\in K[x]$ and two antiderivatives of $h$ only differ by a constant.
	
	Since $W=TU$, we have $\frac{1}{r}R=T^{-1}=\frac{1}{\det(T)}T^*$. So $r$ divides $\det(T)$ and hence $u$ also divides $\det(T)$. {Moreover, $\disc(W)=\disc(U)\det(T)^2$.} {Since $u_i$'s and $\omega_i$'s are integral elements, the discriminants }$ \disc(W), \disc(U)$ are polynomials in $K[x]$. Therefore $\det(T)^2$ divides {the polynomial} $\disc(W)/\disc(U)$, so does $u^2$.
\end{proof}
{If we already know that $W$ is an integral basis, then the quotient $\disc(W)/\disc(U)$ is a unit in $K$. By Lemma~\ref{lem:bound}, we see $u$ is a constant. If no integral basis is available, then from the condition $u^2$ divides $\disc(W)$, an upper bound of $u$ can be chosen as the product over
	$p^{\lfloor r/2\rfloor}$ where $p$ runs through the irreducible factors of
	$\disc(W)$ and $r$ is the multiplicity of $p$ in $\disc(W)$.
	\begin{example}
		Let $m=y^2-x$ and $h=\frac{y}{x}$.
		\begin{enumerate}
			\item For $W=(1,y)$, we have $\disc(U)=4x$. So we can choose $u=1$. In fact, $W$ is an integral basis and $h=(2y)^\prime$.
			\item For $W=(1,(x^2+1)y)$, we have $\disc(W)=4(x^2+1)^2x$. So we can choose $x^2+1$ as an upper bound of $u$. In fact, $h=(\frac{1}{x^2+1}2(x^2+1)y)^\prime$.
		\end{enumerate}
	\end{example}
	If $h$ is integrable in $A$, we shall reduce $h$ to zero, otherwise we hope to remove all possible integrals whose denominators are bounded by $u$. Before that we write $h$ as two parts with denominators $d$ and $e$, respectively. By} the extended Euclidean algorithm, there are { polynomials $r_i,s_i\in K[x]$} such that $h_i=r_id+s_ie$ and $\deg_x(r_i)<\deg_x(d)$. Then the lazy Hermite reduction remainder $h$ decomposes as
\begin{equation}\label{eq:euclidean}
h=\sum_{i=1}^n\frac{h_i}{de}\omega_i=\sum_{i=1}^n\frac{r_i}{d}\omega_i+\sum_{i=1}^n\frac{s_i}{e}\omega_i.
\end{equation}

{Our second goal is to confine the $s_i$'s to a finite-dimensional vector space over~$K$. This is a generalization of the {\em polynomial reduction} in~\cite{chen16a}. In this process, we shall rewrite the second term of the remainder $h$ in~\eqref{eq:euclidean} with respect to another basis. This new basis is used to perform the polynomial reduction and obtain the following additive decomposition.
	\begin{defi}
		Let $f$ be an element in $A$. Let $W$ and $V$ be two $K(x)$-vector space bases of $A$. Let $e,a\in K[x]$ and $M,B\in K[x]^{n\times n}$ be such that $eW^\prime=MW$ and $aV^\prime=BV$. Suppose that $f$ can be decomposed into
		\begin{equation}\label{eq:add decom}
		f=g^\prime+\frac{1}{d}PW+\frac{1}{a}QV,
		\end{equation}
		where $g\in A$, $d\in K[x]$ is squarefree and $\gcd(d,e)=1$, $P,Q\in K[x]^n$ with $\deg_x(P)<\deg_x(d)$ and $Q\in N_V$, which is a finite-dimensional $K$-vector space. The decomposition in ~\eqref{eq:add decom} is called an {\em additive decomposition} of $f$ with respect to $x$ if it satisfies the condition that $P,Q$ are zero if and only if $f$ is integrable in $A$.
	\end{defi}
	Given an algebraic function, its additive decomposition always exists for some integral basis $W$ and some basis $V$ which is a local integral basis at infinity, see~\cite[Theorem 14]{chen16a}. We shall show below that we can always find an additive decomposition with respect to certain suitable bases.}

Let $V$ be a $K(x)$-vector space basis of $A$, and let $a\in K[x]$ and $B=((b_{i,j}))_{i,j=1}^n\in K[x]^{n\times n}$ be such that $aV^\prime=BV$.
We do not require that $\gcd(a,b_{1,1},b_{1,2},\ldots,b_{n,n})=1$. Let $u\in K[x]$ and $p\in K[x]^{n}$. A direct calculation yields that
\begin{align}
\left(\frac{p}{u}V\right)^\prime&=\left(\frac{p}{u}\right)^\prime V+\frac{p}{u}V^\prime=\frac{aup^\prime-au^\prime p+upB}{u^2a}V.
\end{align}
This motivates the following definition.
\begin{defi}
	For a given polynomial $u\in K[x]$, let the map $\phi_V:K[x]^n\to u^{-2}K[x]^n$ be defined by \[\phi_V(p)=\frac{1}{u^2}(aup^\prime-au^\prime p+upB)\] for any $p\in K[x]^n$. We call $\phi_V$ the map for {\em polynomial reduction} with respect to $u$ and $V$, and call the subspace \[\im(\phi_V)=\{\phi_V(p)\mid p\in K[x]^n\}\subseteq u^{-2}K[x]^n\] the {\em subspace for polynomial reduction} with respect to $u$ and~$V$.
\end{defi}
Note that, by construction, if $q=\phi_V(p)$, then  $\frac{q}{a}V=(\frac{p}{u}V)^\prime$. So $\frac{q}{a}V$ is integrable.

We can always view an element of $K[x]^n$ (resp. $K[x]^{n\times n}$) as a polynomial in $x$ with coefficient in $K^n$ (resp. $K^{n\times n}$). In this sense, we use the notation $\lt(\cdot)$ for the leading term of a vector (resp. matrix). For example, if $p\in K[x]^n$ is of the form
\[p=p^{(r)}x^r+\cdots+p^{(1)}x+p^{(0)},\quad p^{(i)}\in K^n,\]
where $p^{(r)}\neq0$, then $\deg_x(p)=r$, $\lt(p)=p^{(r)}x^r$. Let $\{e_1,\ldots,e_n\}$ be the standard basis of $K^n$. Then the $K[x]$-module $K[x]^n$ viewed as $K$-vector space is generated by
\[\S:=\{e_ix^j \mid 1\leq i\leq n, j\in \bN\}.\]

\begin{defi}
	Let $N_V$ be the $K$-subspace of $K[x]^n$ generated by
	\[\{t\in \S\mid t\neq \lt(p) \text{ for all }p\in \im(\phi_V)\cap K[x]^n\}\]
	Then $K[x]^n=(\im(\phi_V)\cap K[x]^n)\oplus N_V$. We call $N_V$ the {\em standard complement} of $\im(\phi_V)$. For any $p\in K[x]^n$, there exists $p_1\in K[x]^n$ and $p_2\in N_V$ such that
	\[\frac{p}{a}V=\left(\frac{p_1}{u}V\right)^\prime + \frac{p_2}{a}V.\]
	This decomposition is called the {\em polynomial reduction of $p$} with respect to $u$ and $V$.
\end{defi}

{If $u=1$ , then the polynomial reduction with respect to $u$ falls back to the situation discussed in~\cite{chen16a}.  }
\begin{prop}
	Let $a\in K[x]$ and $B\in K[x]^{n\times n}$ be such that $aV^\prime=BV$, as before. If $\deg_x(B)\leq\deg_x(a)-1$, then $N_V$ is a finite dimensional $K$-vector space.
\end{prop}
\begin{proof}
	Consider the map $\tilde\phi_V: K[x]^n\to K[x]^n$ defined by \[\tilde\phi_V(p)=aup^\prime-au^\prime p+upB\]
	for any $p\in K[x]^n$.
	Then $\tilde\phi_V(p)=u^2\phi_V(p)$. It is easy to check that $\im(\phi_V)\cap K[x]^n$ and $\im(\tilde\phi_V)\cap u^2 K[x]^n$ are isomorphic as $K$-vector spaces via the multiplication by $u^2$. 
	Considering the codimension of subspaces in $K[x]^n$ over $K$, we have the formula
	\begin{align*}
	\dim_K( N_V)&=\codim_K\left(\im(\phi_V)\cap K[x]^n\right) \\ \nonumber
	&=\codim_K\left( \im(\tilde\phi_V)\cap u^2K[x]^n \right)\\
	&\leq \codim_K\left(\im(\tilde\phi_V)\right)+\codim_K\left( u^2K[x]^n\right).
	\end{align*}
	
	Let $\mu:=\deg_x(a)-1$, $\ell:=\deg_x(u)$ and $s:=\deg_x(p)$. Since $\deg_x(aup^\prime)=\deg_x(au^\prime p)=s+\ell+\mu$, we have \[\deg_x(\tilde\phi_V(p))\leq s+\ell+\max\{\mu,\deg_x(B)\}.\]
	By an argument analogous to \cite[Proposition 12]{chen16a}, we distinguish two cases $\deg_x(B)<\mu$ and $\deg_x(B)=\mu$, and get the codimension of $\im(\tilde\phi_V)$ is finite. Since $K[x]^n/(u^2K[x]^n)\cong (K[x]/u^2K[x])^n$, the codimension of $u^2K[x]^n$ is also finite.
\end{proof}

The condition $\deg_x(B)\leq\deg_x(a)-1$ is satisfied if $V$ is a local integral basis at infinity, but may not hold for an arbitrary basis. So we introduce a weaker basis that satisfies the degree condition. This is an analogue of suitable basis at infinity.

\begin{defi}
	A basis $V$ of $A$ is called \emph{suitable at infinity}
	if for $a\in K[x]$ and $B\in K[x]^{n\times n}$ such that $aV'=BV$ we have
	$\deg_x(B)<\deg_x(a)$.
\end{defi}

There always exists a basis which is suitable at infinity. We can find such a basis as follows.
Start from an arbitrary $K(x)$-basis $V=(v_1,\dots,v_n)$ of the function field~$A$.
We can make its elements $v_1,\dots,v_n$ integral at infinity by replacing each $v_i$ by $x^{-\tau} v_i$
for a sufficiently large $\tau\in\set N$.
Consider $a\in K[x]$ and $B\in K[x]^{n\times n}$ be such that $aV'=BV$.
If $\deg_x(B)<\deg_x(a)$, we are done.
If not, consider a row $b$ in $B$ with $\deg_x(b)\geq\deg_x(a)$ and set $v=x a^{-1} bV$.
Then $v$ is integral at infinity (because at infinity differentiating increases the valuation)
but it does not belong to the $K(x)_\infty$-module generated by~$V$ (because of $\deg_x(b)\geq\deg_x(a)$).
Therefore the $K(x)_\infty$-module generated by $V$ and $v$ is strictly larger than the
$K(x)_\infty$-module generated by~$V$. Replace $V$ by a basis of this enlarged module, and
update $a$ and $B$ such that $aV'=BV$. If we now have $\deg_x(B)<\deg_x(a)$, we are done,
otherwise repeat the process just described.

The iteration will terminate because with every update of $V$ the $K(x)_\infty$-module generated
by it gets enlarged, and since all these modules are contained in the module of elements of $A$
that are integral at infinity, after at most finitely many updates $V$ will be a local integral
basis at infinity. At least then, the desired degree condition must hold, because otherwise
there would be an integral element which is not a $K(x)_\infty$-linear combination of the basis
elements, in contradiction to the basis being integral at infinity.

\begin{theorem}\label{thm:add decomp}
	Let $f$ be an element in $A$. Let $V$ be a basis of $A$ which is suitable at infinity. Then there exists a suitable basis $W$ of $A$ such that $f$ admits an additive decomposition { in~\eqref{eq:add decom} with respect to the bases }$W$ and $V$.
\end{theorem}
\begin{proof}
	We present a constructive proof to show the existence of an additive decomposition of $f$. After performing the lazy Hermite reduction on $f$, we get
	\[f=\tilde g^\prime+\frac{1}{d}PW+\frac{1}{e}UW\]
	where $W$ is a suitable basis, $P=(r_1,\ldots,r_n)\in K[x]^n$ and $U=(s_1,\ldots,s_n)\in K[x]^n$ with $r_i,s_i$ introduced in~\eqref{eq:euclidean}. Let $W=\frac{1}{b}CV$ for some $b\in K[x]$ and $C\in K[x]^{n\times n}$. Let $a\in K[x]$ and $B\in K[x]^{n\times n}$ be such that $aV^\prime=BV$. Multiplying $a$ and $B$ by some polynomial, we may assume that $a$ is a common multiple of $e$ and $b$. Rewriting the remainder in terms of the new basis $V$, we get
	\begin{equation}\label{eq: decom 1}
	\frac{1}{e}UW=\frac{1}{eb}UCV=\frac{1}{a}\tilde UV,
	\end{equation}
	for some $\tilde U\in K[x]^n$. By Theorem~\ref{lem:bound}, if $f$ is integrable, there exist $\tilde u\in K[x]$ and $R\in K[x]^n$ such that
	\begin{equation}\label{eq:decom 2}
	\frac{1}{e}UW=\left(\frac{RW}{\tilde u}\right)^\prime=\left(\frac{RCV}{ u}\right)^\prime,
	\end{equation}
	where $u=\tilde u b$. Next, we apply the polynomial reduction with respect to the polynomial $u$ and decompose $\tilde U$ into $\phi_V(\tilde U_1)+\tilde U_2$ with $\tilde U_1,\tilde U_2 \in K[x]^n$ and $\tilde U_2\in N_V$. Then we have
	\begin{equation}\label{eq:decom 3}
	\frac{1}{a}\tilde UV=\left(\frac{1}{u}\tilde U_1V\right)^\prime+\frac{1}{a}\tilde U_2V.
	\end{equation}
	We then get the decomposition~\eqref{eq:add decom} by setting $g=\tilde g+\frac{1}{u}\tilde U_1V$ and $Q=\tilde U_2$.
	
	Assume that $f$ is integrable. Then Theorem~\ref{lem:bound} implies that $d\in K$. Since $\deg_x(P)<\deg_x(d)$, we have $P=0$. Combining \eqref{eq:decom 3}, \eqref{eq: decom 1} and~\eqref{eq:decom 2} yields that
	\begin{equation}
	\frac{1}{a} Q V= \frac{1}{a}\tilde UV-\left(\frac{1}{u}\tilde U_1V\right)^\prime=\left(\frac{1}{u}\tilde QV\right)^\prime,
	\end{equation}
	where $\tilde Q=RC-\tilde U_1$. So $Q=\phi_V(\tilde Q)\in \im(\phi_V)\cap K[x]^n$. Since $\im(\phi_V)\cap K[x]\cap N_V=\{0\}$, it follows that $Q=0$.
\end{proof}
\begin{example}
	We continue with Example~\ref{ex:doulbe root at infy} by applying the polynomial reduction to the lazy Hermite remainder $h=\frac{1}{3x(x^2+1)}(x^2+1)y$. Since $\disc(W)=4(x^2+1)^2x$, we choose $u=x^2+1$. Note that $W$ is already a suitable basis at infinity. The map for the polynomial reduction with respect to $W$ and $u$ is $\phi(p)=\frac{1}{u^2}(eup^\prime-eu^\prime p+epM)$ for any $p\in K[x]^n$. Then $h=(\frac{1}{3(x^2+1)}2(x^2+1)y)^\prime$ reduces to~$0$.
\end{example}

\section{Reduction-Based Telescoping} \label{SECT:tele}

Lazy Hermite reduction in combination with the polynomial reduction just described can be used for deciding
whether a given algebraic function admits an algebraic integral.
Most algebraic functions don't.
The next question of interest may then be whether the algebraic function at hand can be deformed in some way
to a related one that is algebraically integrable.
Creative telescoping produces such a deformation.
It is applicable to functions~$f$ which besides the integration variable $x$ involve some other parameter~$t$.
The task of creative telescoping is to find a nonzero operator $L(t,D_t)$ such that $L(t,D_t)\cdot f$ is integrable.
Such an operator is called a \emph{telescoper} for~$f$.

Several techniques are known for computing such a telescoper. The so-called reduction based approach is one
of them, and it has attracted a lot of attention in recent years. It was first proposed for rational functions~\cite{bostan10b}.
Given $f=p/q\in K(x)$ with $K=C(t)$, we can use Hermite reduction to find $g_i,h_i\in K(x)$ ($i=0,1,2,\dots$) such that
$D_t^i f = D_x g_i + h_i$. If $c_0,\dots,c_r\in K$ are such that $c_0 h_0+\cdots+c_rh_r=0$, then
\[
(c_0+\cdots +c_rD_t^r)\cdot f = D_x\cdot( c_0 g_0 + \cdots + c_r g_r),
\]
so the operator $L=c_0+\cdots+c_rD_t^r$ is a telescoper for~$f$.

Some recomputation can be avoided by observing that instead of $D_t^i f$ we may as well integrate $D_t h_{i-1}$,
because $D_x$ and $D_t$ commute. With this optimization, reduction based telescoping can be summarized as follows.

\begin{algorithm}
	Input: a function $f$ depending on $x$ and~$t$;\\
	Output: a telescoper for $f$
	
	\step 11 find $g_0,h_0$ such that $f=g_0'+h_0$.
	\step 21 for $r=1,2,3,\dots$ do:
	\step 32 if $h_0,\dots,h_{r-1}$ are linearly dependent over $K$
	\step 43 return $c_0+c_1D_t+\cdots+c_{r-1}D_t^{r-1}$ with $c_i$ not all zero and $\sum_{i=0}^{r-1} c_i h_i = 0$.
	\step 52 find $g_r,h_r$ such that $D_t h_{r-1}=g_r'+h_r$
\end{algorithm}

The termination of this procedure can be secured in two ways.
One way is to ensure that the remainders $h_0,h_1,\dots$ belong to a finite dimensional $K$-vector space.
Then the remainders must eventually become linearly dependent.
The second way is to ensure that the map which sends $f$ to $h$ such that $f=g'+h$ for some $g$ is $K$-linear
and has the set of all integrable elements as its kernel.
It is then guaranteed that the procedure will not miss a telescoper, so it will terminate because we know that
for every algebraic function there does exist a telescoper.

Besides for rational functions, both arguments have been worked out for various larger classes of
functions~\cite{bostan13a,chen15a,chen17a,bostan18a,hoeven20},
including the class of algebraic functions~\cite{chen16a}. The version for algebraic functions is uses Trager's
Hermite reduction followed by a polynomial reduction, both steps requiring an integral basis of the function
field. Using {Theorem}~\ref{thm:add decomp}, we will argue that reduction based telescoping also works with lazy Hermite reduction
and the variant of polynomial reduction developed in the previous section, with the obvious advantage that no
integral bases computation is needed.

For doing so, we must take into consideration that lazy Hermite reduction takes a suitable basis as input but
may return the result with respect to an adjusted suitable basis. Let $W_0,W_1,\dots$ denote the suitable
bases with respect to which the result of the $i$th call to lazy Hermite reduction is returned. By supplying
$W_{i-1}$ as input to the $i$th call, we can ensure that the $K[x]$-module generated by $W_{i-1}$ is
contained in the $K[x]$-module generated by~$W_i$, for every~$i$. Therefore, if $W_r\neq W_{r-1}$ for
some~$r$, we can rewrite all remainders $h_0,\dots,h_{r-1}$ in the bases $W_r$ and $V$ without introducing
new denominators. (Note that no update is required for~$V$.)
In order to meet the conditions specified in {Theorem}~\ref{thm:add decomp}, it may be necessary to rerun
the polynomial reduction on the new representations of the old remainders. The termination of the algorithm
then follows via {the second way} indicated above.
In order to also justify termination in the first way, it suffices to observe that we can keep $V$ fixed
throughout the computation, so the termination follows directly from the finite dimension of~$N_V$.



\section{Experiments} \label{SECT:experiments}

For {the paper~\cite{chen12d} in 2012}, we have created a collection of about 100
integration problems, mostly originating from applications in
combinatorics~\cite{pemantle08,bostan16b}, and we have compared the performance
of six different approaches, including Chyzak's
algorithm~\cite{chyzak97,chyzak00,koutschan13} as implemented by
Koutschan~\cite{koutschan09,koutschan10c}, Koutschan's ansatz-based
method~\cite{koutschan10b}, as well as the method based on residue elimination
we proposed in~\cite{chen12d}. The result of the evaluation was somewhat
inconclusive. Most algorithms outperformed the other algorithms at least
for some examples. At the same time, the timing differences can be significant.

For the present paper, we have evaluated the performance of reduction based
creative telescoping using lazy Hermite reduction on the benchmark set from
2012. The runtime was taken on the same computer (a 64bit Linux server with 24
cores running at 3GHz; in 2012 it had 100G RAM, meanwhile it was upgraded to
700G).  The new experiments were performed with a more recent version of
Mathematica (Mathematica 12.1.1). Timings and code are available on our
website~\cite{website}.

An interesting observation is that in all examples of the collection (at least
those which finished within the specified time limit of 30h), the lazy Hermite
reduction procedure never encountered a situation where the linear
system~\eqref{eq:sys} did not have a unique solution so that a basis change
would have been required. In almost all cases, we also found that no basis
change was needed before entering the reduction process in order to turn to a
suitable basis. Examples where the default basis was not suitable and needed to
be adjusted generally took much more time than examples where the default basis
was already suitable.

In comparison to the earlier methods, the inclusion of reduction based creative
telescoping does not change the general observation that no algorithm is clearly
superior to the others. The new approach is faster than the other approaches on
some of the examples, while on others it is much worse.

\begin{example}
	Consider the rational function
	\[
	f(x,y,t)=(1-x-y-t+\tfrac34(xy+xt+yt))^{-1}.
	\]
	The computation of an annihilating operator for $\res_{x,y} f(x,y,\frac{t}{x^3y})$ is
	completed in about 4.5 seconds by our new approach. The other techniques need at least
	twice as long and up to 500 seconds.
	In this example, the telescoper has order~6 and coefficients of degree~15.
	On the other hand, computing an annihilating operator for $\res_{x,y} f(x,y,\frac{t}{x^2y^2})$
	takes about half an hour using the reduction based approach, while all other approaches
	all need less than 90 seconds.
	Here the telescoper has order~4 and coefficients of degree~11.
	It is not clear why the algorithms perform so differently on such closely related input.
\end{example}

In view of this diverse and unpredictable performance of the various algorithms, it is
advantageous to have several independent techniques available. Having more techniques at our
disposal increases the chances that a hard instance arising from an application can be
completed by at least one of them.


\bibliographystyle{plain}
\bibliography{integral}

\end{document}